\documentclass[journal]{IEEEtran}

\usepackage{inputenc}
\usepackage{mathtools}
\usepackage{color}
\usepackage{array}
\usepackage{verbatim}
\usepackage{float}
\usepackage{amsmath}
\usepackage{amsthm}
\usepackage{amssymb}
\usepackage{graphicx}
\usepackage{longtable}
\usepackage{multirow}
\usepackage{balance}
\usepackage{mathtools}

\newcommand\myeq{\stackrel{\mathclap{\normalfont\mbox{def}}}{=}}

\usepackage[unicode=true,
bookmarks=false,
breaklinks=false,pdfborder={0 0 1},colorlinks=false]
{hyperref}
\hypersetup{
	colorlinks,bookmarksopen,bookmarksnumbered,citecolor=red,urlcolor=red}
\usepackage{cite}

\usepackage{indentfirst}
\usepackage{textcomp}
\usepackage[T1]{fontenc}
\usepackage{amsthm}
\usepackage{amsmath,amstext,amsfonts}
\usepackage{amssymb}
\usepackage{mathtools}
\usepackage[cmintegrals]{newtxmath}
\usepackage{bm}  
\usepackage{siunitx}
\usepackage{physics}
\DeclareMathOperator*{\argmin}{arg\,min}

\usepackage{graphicx}
\graphicspath{%
	{Figures/}
	{Figures/TiKz/}
	{Figures/Ipe/}
	{Figures/XFig/}
	{Figures/Inkscape/}
	{Figures/Matlab/}
	{Figures/Scanned/}
	{Figures/Saved/}
}
\usepackage{tikz}
\usepackage[list=true,labelformat=simple]{subcaption}
\usepackage{cite}
\usepackage{url}
\usepackage{xcolor}
\usepackage{siunitx}

\usepackage[inline]{enumitem}

\usepackage{algorithm,algpseudocode,float}
\usepackage{setspace} 

\algnewcommand{\algorithmicand}{\textbf{ and }}
\algnewcommand{\algorithmicor}{\textbf{ or }}
\algnewcommand{\AlgAnd}{\algorithmicand}
\algnewcommand{\AlgOr}{\algorithmicor}

\renewcommand{\qedsymbol}{$\blacksquare$}

\usepackage[noabbrev]{cleveref}  
\Crefname{figure}{Fig.}{Figs.}

\usepackage{tabularx,booktabs}
\newcolumntype{C}{>{\centering\arraybackslash}X} 
\usepackage{lipsum}

\makeatletter
\let\oldforeign@language\foreign@language
\DeclareRobustCommand{\foreign@language}[1]{%
	\lowercase{\oldforeign@language{#1}}}

\renewcommand{\qedsymbol}{$\blacksquare$}

\floatstyle{ruled}
\newfloat{algorithm}{tbp}{loa}
\providecommand{\algorithmname}{Algorithm}
\floatname{algorithm}{\protect\algorithmname}

\let\oldforeign@language\foreign@language
\DeclareRobustCommand{\foreign@language}[1]{%
	\lowercase{\oldforeign@language{#1}}}

\ifCLASSINFOpdf
\else
\fi

\newtheorem{defn}{Definition}

\newtheorem{lem}{Lemma}

\newtheorem{thm}{Theorem}
\newtheorem{rem}{Remark}

\newtheorem{assum}{Assumption}
\begin{document}
	
	
	\title{An Online Model-Following Projection Mechanism Using Reinforcement Learning}

	\author{Mohammed I. Abouheaf$^1$, Hashim A. Hashim$^2$, Mohammad A. Mayyas$^1$, and Kyriakos G. Vamvoudakis$^3$
		\thanks{$^1$M. I. Abouheaf and M. A. Mayyas are with the College of Technology, Architecture \& Applied Engineering, Bowling Green State University, Bowling Green, OH, 43403, USA, email: \{mabouhe,mmayyas\}@bgsu.edu.
			
			$^2$H. A. Hashim is with the Department of Mechanical and Aerospace Engineering, Carleton University, Ottawa, Ontario, K1S-5B6, Canada, e-mail: Hashim.Mohamed@carleton.ca.
			
			$^3$K. G. Vamvoudakis is with the Daniel Guggenheim 	School of Aerospace Engineering, Georgia Institute of Technology, Atlanta, GA, 30332, USA, e-mail: kyriakos@gatech.edu.}
		\thanks{
			This work was supported in part by the National Sciences and Engineering
			Research Council of Canada (NSERC), under the grants RGPIN-$2022-04937$ and by the National Science Foundation (NSF) under grant Nos. S\&AS-$1849264$, CPS-$1851588$, CPS-$2227185$, and CPS-$2038589$. 
	}}


	\maketitle
	\pagestyle{empty}
	\thispagestyle{empty}

	\begin{abstract}
		%
		In this paper, we propose a model-free adaptive learning solution for a model-following control problem. This approach employs policy iteration, to find an optimal adaptive control solution. It utilizes a moving finite-horizon of model-following error measurements.
		In addition, the control strategy is designed by using a projection mechanism that employs Lagrange dynamics. It allows for real-time tuning of  derived actor-critic structures to find the optimal model-following strategy and sustain optimized adaptation performance. Finally, the efficacy of the proposed framework is emphasized through a comparison with sliding mode and high-order model-free adaptive control approaches. 
	\end{abstract}
	
	\IEEEpeerreviewmaketitle{}

	\section{Introduction}
	Model Reference Adaptive Systems (MRASs) are utilized in many applications such as the actuation of manipulators, guidance of unmanned vehicles, and motion planning~\cite{Chen2021,MPC2021,Liu2018,Hol2002,Kam1,Byrne1995}. The optimal tracking control solutions are mostly implemented offline and  require partial or complete knowledge of the physical models as well as the desired reference trajectories~\cite{Lewis2012,aastrom2013adaptive,DMMPC2021,Hol2002,Kam1,Chen2021,MPC2014,MPC2021,AUV1,Robust2019,Byrne1995,Moore2014,Liu2018,SHI2017,netw1}.  
	%
	%
	%
	For nonlinear systems with internal passivity, the model-reference tracking problem is solved using sliding mode surfaces along with a velocity observer~\cite{Robust2019}. Nonetheless, the control strategy partially relied on the process dynamics, where a zero-state detectability condition is considered to guarantee asymptotic stability of the equilibrium point.  For underactuated nonlinear systems of moderate order, \cite{Moore2014} proposed an MRAS solution that adopted the concept of the sum-of-squares polynomial optimization. The derived strategy is partially reliant on the process dynamics. The same is true for the solution presented in~\cite{RMPC2021}, where a robust model predictive control approach is considered forthe reference batch processes.
	MRAS solutions based on graphical games have been developed for multi-agent systems~\cite{AbouheafCTT2015,AbouheafICRA19,AbouheafAuto14}. These solutions are designed for linear time-invariant systems and require partial knowledge of each agent's dynamics to derive local strategies. As such, these approaches do not address the nonlinearity of the agents.  
	
	Reinforcement Learning (RL) mechanisms have not been fully investigated to develop model-following adaptive control strategies~\cite{MFAC2021}. RL is concerned with guiding the agent towards the best strategies after interactions with the environment to maximize (minimize) a cumulative reward (cost) ~\cite{sut92,Sutton_1998,Bertsekas1996}. RL solutions can be found using several techniques including two-step mechanisms such as policy iteration (PI).
	%
	PI solution evaluates and improves a given strategy in an iterative manner~\cite{Bertsekas1996,AbouheafACC19,Busoniu2010}. The evaluation of the policy can be done using approaches such as  
	least squares (LS) and recursive LS (RLS)~\cite{Busoniu2010, Srivastava2019}. 
	An off-policy RL approach is considered to solve the Algebraic Riccati Equation (ARE) in~\cite{Bah2017}. Another PI mechanism is adopted to solvethe output-based regulation of a cooperative control problem in~\cite{Lewis20}. On the other hand, approximation tools such as the means of adaptive critics are adopted to implement the RL solutions. The adaptive critic is a device that learns to anticipate reinforcing events in a way that makes it a useful conjunct to another component, the actor, that adjusts behavior to maximize the frequency and/or magnitude of reinforcing events~\cite{Sutton2008,CrtBahare,CrtLewis,CrtZhao}. Gradient approaches are used to tune the actor and critic weights. The adaptation schemes can vary depending on the desired function approximation structure and the underlying solutions employ supporting conditions such as in~\cite{MFAC2021}, where a pre-designed strategy that requires resetting conditions is considered. In~\cite{Bahare14}, the tracking problem is solved for linear-time invariant systems, where it employed a Q-learning method for an overall augmented system. 
	Another approximate model-free approach based on adaptive critics is adopted in~\cite{AbouheafTrans20} to control a flexible wing aircraft. Nonetheless, the feedback strategy relied on a non-optimal guidance vector embedded within that strategy. 
	%

	This discussion about the challenges associated with several MRAS solutions that exist in the literature motivates us to develop a model-following strategy with the following properties: 1) ease of implementation in a digital environment such as microprocessors, 2) ability to utilize measurements of the process without incorporating any explicit dynamical information in the underlying strategy, 3) capability to solve model-following problems with high-order error dynamics using feasible adaptive strategies, and 4) enabling simultaneous multi-objective optimization of the model-following and strategy adaptation performances. 
	
	\paragraph*{Contributions} The contributions of the work are threefold. First it formulates a novel model-following adaptive learning solution that  requires only the real-time measurements of the process as inputs to the control strategy. Second, the framework is flexible with respect to the order of the model-following error dynamics, and finally uses a novel projection mechanism based on Lagrange dynamics to adapt the gains of the control strategy.
	
	\textit{Mathematical notation:} The following notation and definitions are adopted by the mathematical setup of the adaptive learning solution. $\mathbb{R}, \, \mathbb{N}, \text{and } \mathbb{Z}_{0}^{+}$ refer to the sets of real numbers, positive whole numbers, and non-negative integers, respectively. $\norm{.}_1$ is the 1-norm of a vector. $\grad{g}$ is a gradient of $g$. {\tiny $\bigotimes$} signifies a Kronecker product. Let the $\ell_\infty-$ norm of a sequence $\{ {\varsigma(e)} \}_{e=0}^{\infty}$ be given by $ \lVert\varsigma \rVert_{\infty} = \sup\limits_{e \in \mathbb{N}} \lVert \varsigma(e) \rVert_\infty$ with $\ell_\infty \myeq \{\varsigma: \lVert\varsigma \rVert_{\infty} < \infty\}$ and $\ell_2 \myeq \{\varsigma: \lVert\varsigma \rVert_{2} < \infty\}$.
	
	\paragraph*{Structure} The remainder of the paper is organized as follows. Section~\ref{sec:Preliminaries} introduces the mathematical setup of the model-following control problem. Moreover, the duality between the Hamilton-Jacobi-Bellman (HJB) and Bellman optimality equations is explained. This is needed to develop a temporal difference mechanism and to derive an optimal model-following strategy. Then, a temporal difference solution based on PI is introduced in Section~\ref{sec:RL}. Further, the convergence conditions of the PI solution are discussed. Section~\ref{sec:crit} presents the actor-critic approximation mechanism of the model-free RL solution. This is done using a projection approach that is based on Lagrange dynamics to guarantee convergence of the adapted actor-critic weights. The solution is validated using nonlinear and linear systems with state and input delays in Section~\ref{sec:Sim}. Finally, Section~\ref{sec:conclus} concludes the work.

	\section{Problem Formulation\label{sec:Preliminaries}}
	This section lays out the mathematical foundation of the adaptive solution using optimal control theory ~\cite{Lewis2012}. 
	Consider a discrete-time nonlinear system described by
	\begin{equation}
		{\bf X}_{k+1}={f_k({\bf X}_k,{\bf \, u}_k)} \text{ and } {\bf y}_k= g_k\left({\bf X}_k\right),\ k\in \mathbb{N}
		\label{eq:dyn}
	\end{equation}
	where ${\bf X}_k \in \mathbb{R}^{n},$ ${\bf \, u}_k \in \mathbb{R}^{t},$ and ${\bf y}_k \in \mathbb{R}^{p}$ represent the state vector, control vector, and output vector, respectively. Consider that the model-to-follow dynamics are given $\forall k\in\mathbb{N}$ by
	\begin{equation}
		{\bf X}^m_{k+1}={f^m_k({\bf X}^m_k,{\bf \, u}^m_k)} \text{ and }
		{\bf y}^m_k= g^m_k\left({\bf X}_k\right)
		\label{eq:mod_dyn}
	\end{equation}
	with ${\bf X}^m_k \in \mathbb{R}^{v},$ ${\bf \, u}^m_k \in \mathbb{R}^{m},$ and ${\bf y}^m_k \in \mathbb{R}^{p}$ denoting the state vector, control vector, and output vector, respectively.
	The problem can be considered as an optimal regulation of the difference between the output of the process and that of the reference model as shown in Fig.~\ref{fig:fig}, i.e., $\lim\limits_{k \rightarrow \infty}\lVert{\bf \varepsilon}_k\rVert\rightarrow \bf 0,$ ${\bf\varepsilon}_k:={\bf y}^m_k-{\bf y}_k$.
	
	The control signal is given by ${\bf u}_{k+1}={\bf u}_{k}+{\bf \mu}^{\bf \omega}_{k},$ where ${\bf \mu}^{\bf \omega}_k$ is a correction control signal that is decided using a real-time adaptive strategy ${\bf \omega}\in \mathbb{R}^{t \times (\left(r+1\right)\times p)}$. The signal ${\bf \mu}^{\bf \omega}_k={\bf \omega} \, \mathcal{E}_k$ employs a flexible-size error vector defined by $\mathcal{E}_k=[\begin{array}{ccccc}
		{\bf\varepsilon}^\textrm{T}_k & {\bf\varepsilon}^\textrm{T}_{k-1} & {\bf\varepsilon}^\textrm{T}_{k-2} & \dots&{\bf\varepsilon}^\textrm{T}_{k-r}\end{array}]^{\textrm{T}}\in\mathbb{R}^{\left(r+1\right)\times p}$. The number of employed error samples reflects the order of the model-following error dynamics.
	Thus, the ultimate goals are to design a framework to avoid the solutions of a set of coupled difference equations backwards in time as well as to develop computationally efficient strategies.

	%
	
	%
	\begin{figure}[h]
		\centering
		\includegraphics[scale=0.45]{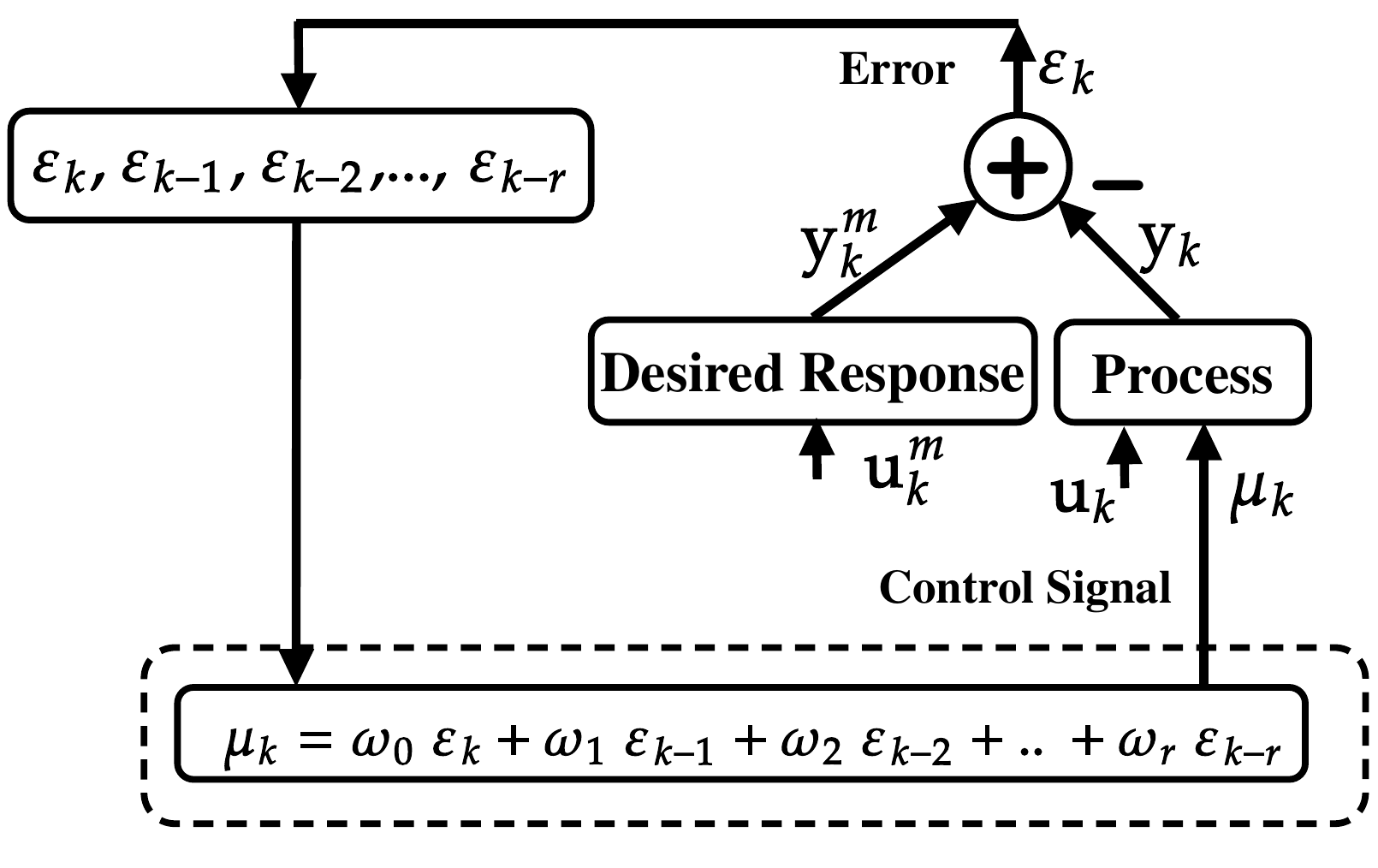}
		\caption{Model-reference adaptive system.}
		\label{fig:fig}
	\end{figure}
	
	%
	
	A user-defined cost functional will be used to measure the quality of a strategy ${\bf \omega}$ such that 
		$\mathcal{U}_k\, \left({\bf \mathcal{E}}_k,{\bf \mu}^{\bf \omega}_k\right)=\frac{1}{2}\left({\bf \mathcal{E}}_k^\textrm{T} \, {\bf \mathcal Q} \, {\bf \mathcal{E}}_k+ {{\bf \mu}^{\bf \omega}_k}^\textrm{T} \, {\bf \mathcal  R}\,  {\bf \mu}^{\bf \omega}_k \right),$ where $  {\bf \mathcal Q}\succeq 0 \in\mathbb{R}^{((r+1) p) \times ((r + 1) p)}$ and $  {\bf \mathcal R} \succ 0  \in\mathbb{R}^{t \times t}$ are symmetric weighting matrices. The cost function is quadratic and convex in the error vector ${\bf \mathcal{E}}_k$ and control strategy ${\bf \mu}^{\bf \omega}_k$. The overall performance of a control strategy ${\bf \omega}$ is evaluated using a performance index given by
	\begin{equation}
		J_k^{\bf \omega}=\sum_{i=k}^{\infty}\mathcal{U}_i \, \left(\mathcal{E}_i,{\mu}^{\bf \omega}_i\right).
		\label{eq:index}
	\end{equation}
	\par	The optimal strategy ${\bf \omega}^o$ is decided by solving the HJB equation of the model-following control problem. The structure of the convex cost function motivates the form of the control strategy to be linear in the error vector ${\bf \mathcal{E}}_k$  i.e., the model-following error dynamics. The Hamiltonian is then given $\forall \mathcal{E}_k, {\bf \lambda}_{k+1}, {\bf \mu}_k$ by
	\begin{equation}
		H(\mathcal{E}_k, {\bf \lambda}_{k+1}, {\bf \mu}_k)= {\bf \lambda}^\textrm{T}_{k+1} {f^{\varepsilon}_k(\mathcal{E}_k,{\bf \mu}_k)}+\mathcal{U}_k \, \left(\mathcal{E}_k,{\mu}_k\right),
		\label{eq:ham}
	\end{equation}
	where ${f^{\varepsilon}_k(\mathcal{E}_k,{\bf \mu}_k)}$ is a constraint that is dictated by the model-following error dynamics.
	Thus, the optimization problem shall find the optimal strategy ${\bf \omega}^o$ while satisfying the constraint 
	${f^{\varepsilon}_k(\mathcal{E}_k,{\bf \mu}_k)}={\bf Z}_{k+1},$
	where $\small {\bf Z}_{k+1}= \left[
	\begin{array}{cc}
		\mathcal{E}^\textrm{T}_{k+1}
		&
		{\bf \mu}^\textrm{T}_{k+1}
	\end{array}
	\right]^\textrm{T}$ or equivalently $\small {\bf Z}_{k+1}= \left[
	\begin{array}{cc}
		\mathcal{E}^\textrm{T}_{k+1}
		&
		{\bf \mathcal{E}}^\textrm{T}_{k+1} {\bf \omega}^\textrm{T}
	\end{array}
	\right]^\textrm{T}$.

	\begin{rem}
		\label{rem:rmk1}
		This development enables a flexible strategy that is scalable in terms of the number ofthe model-following error measurements. Hence, the desired strategy takes the form of a digital PID controller but with adaptable gains in real-time. This strategy mimics to some extent a gain scheduler but with real-time adaptation capabilities. Furthermore, this is useful in the case of optimizing numerous coupled model-following loops acting simultaneously.\hfill  $\square$
	\end{rem}

	\begin{assum}
		\label{Asm:1}
		Given the desired response of the model-to-follow dictated by \eqref{eq:mod_dyn}, assume that \eqref{eq:dyn} is stabilizable around the reference-trajectory ${\bf y}^m_k, \forall k\in \mathbb{N}$.  \hfill  $\square$
	\end{assum}
	\begin{assum}
		\label{Asm:2}
		There exists a strategy ${\bf \omega}$ such that the control signals ${\bf \mu}_k={\bf \omega}\,\mathcal{E}_k, \forall k\in\mathbb{N}$  stabilize the system \eqref{eq:dyn} around the reference-trajectory \eqref{eq:mod_dyn}.\hfill  $\square$
	\end{assum}

	The following result explains the duality between the Hamiltonian and the temporal difference (TD) or Bellman equation using the discrete-time Hamilton-Jacobi (DTHJ) theory. This result is necessary to develop our adaptive solution.	
	\begin{thm}
		\label{thm:Theorem1}
		Let the value function $V({\bf Z}_{k})>0$ be quadratic and convex in the regulation error vector $\mathcal{E}_k$ with $V({\bf 0})=0$. Then,
		\begin{enumerate}
			\item[a.] The value function $V({\bf Z}_{k})$ satisfies the DTHJ equation given by
			\begin{equation}
				H(\mathcal{E}_k, \displaystyle \frac{\partial V({\bf Z}_{k+1})}{\partial {\bf Z}_{k+1}}, {\bf \mu}_k)=0,\  k\in \mathbb{N}.
				\label{eq:HJ}
			\end{equation}
			\item[b.] The value function $V({\bf Z}_{k})$ represents a Lyapunov function. 
		\end{enumerate}
		
	\end{thm}
	
	\begin{proof}\hspace{200pt}\newline
		\indent \textit{a.} The value function $V({\bf Z}_{k})=\frac{1}{2}\left( {\bf Z}^\textrm{T}_{k} \, {\bf S } \,  {\bf Z}_{k}\right)$ is convex in the error vector $\mathcal{E}_k$ and thus it serves as a Lyapunov candidate where ${\bf S}={\bf S}^\textrm{T}\succ 0 \in\mathbb{R}^{((r + 1) p+t)  ((r+ 1) p+t)}$ and it has a matrix-block structure defined by $\left[
		\begin{array}{cc}
			{\bf S}_{\mathcal{E}\mathcal{E}} & {\bf S}_{\mathcal{E}{\bf \mu}} \\
			{\bf S}_{{\bf \mu} \mathcal{E}} & {\bf S}_{{\bf \mu}{\bf \mu}}
		\end{array}
		\right]$.
		
		The value function can be expressed in terms of the performance index $J_k$ such that $
		V({\bf Z}_{k}) \, \triangleq \, J_k=\sum_{i=k}^{\infty}\mathcal{U}_i\left(\mathcal{E}_i,{\mu}_i\right)
		$
		\noindent or equivalently given by
		\begin{equation}
			V({\bf Z}_{k})=\sum_{i=k}^{\infty}\mathcal{U}_i+{\bf \lambda}^\textrm{T}_{i+1} \left({\bf Z}_{i+1}-{f^{\varepsilon}_i(\mathcal{E}_i,{\bf \mu}_i)}\right).
			\label{eq:Bell}
		\end{equation}
		The Hamiltonian \eqref{eq:ham} and the value function \eqref{eq:Bell} yield  
		\begin{equation}
			{\bf \lambda}^\textrm{T}_{k+1} {\bf Z}_{k+1}+V({\bf Z}_{k})-V({\bf Z}_{k+1})-H(\mathcal{E}_k, {\bf \lambda}_{k+1}, {\bf \mu}_k)=0.
			\label{eq:Ham1}
		\end{equation}
		Taking the gradient of \eqref{eq:Ham1} with respect to ${\bf Z}_{k+1}$ yields
		\begin{align*}
			\frac{\partial{\bf \lambda}_{k+1}}{\partial {\bf Z}_{k+1}}^\textrm{T}
			\left({\bf Z}_{k+1}-\frac{\partial H(\mathcal{E}_k, {\bf \lambda}_{k+1}, {\bf \mu}_k)}{\partial {\bf \lambda}_{k+1}}\right)+{\bf \lambda}_{k+1}-\frac{\partial V({\bf Z}_{k+1})}{\partial {\bf Z}_{k+1}}=0,
		\end{align*}
		i.e., ${\bf \lambda}_{k+1}=\frac{\partial V({\bf Z}_{k+1})}{\partial {\bf Z}_{k+1}}$. Therefore, the value function satisfies \eqref{eq:HJ}.
		\newline
		\indent \textit{b.} Since, $V({\bf Z}_{k+1})-V({\bf Z}_{k})=-{\cal U}_k \, \left(\mathcal{E}_k,{\mu}_k\right) \le {\bf 0}$, then the candidate structure $V({\bf Z}_{k})$ is a Lyapunov function. 
	\end{proof}
	The optimal solution is found by solving the HJB equation 
	\begin{equation}
		H(\mathcal{E}_k, {\bf \nabla} V^o_{k+1}, {\bf \mu}^o_k)=0,
		\label{eq:HJB}
	\end{equation}
	where ${\bf \mu}^o_k$ refers to the optimal correction control signal (i.e., ${\bf \mu}^o_k= \argmin\limits_{{\bf \mu}_k} H(\cdot)={\bf \omega}^o\,\mathcal{E}_k$), ${\bf \omega}^o$ is the optimal strategy, $V^o$ is the optimal value function, and $\displaystyle {\bf \nabla} V^o_{k+1}={\partial V^o({\bf Z}_{k+1})}/{\partial {\bf Z}_{k+1}}$. The relation between Lagrange multiplier ${\bf \lambda}_{k+1}$ and the gradient $\displaystyle {\bf \nabla} V_{k+1}$ yields the following Bellman equation
	\begin{equation}
		V^o({\bf Z}_{k})={\cal U}_k \, \left(\mathcal{E}_k,{\mu}^o_k\right)+V^o({\bf Z}_{k+1}),
		\label{eq:Bello}
	\end{equation}
	where the optimal strategy is calculated such that
	${\bf \mu}^o_k= \argmin\limits_{{\bf \mu}_k} H\left(\mathcal{E}_k, {\bf \nabla} V_{k+1}, {\bf \mu}_k\right)= \argmin\limits_{{\bf \mu}_k} V({\bf Z}_{k})$. This strategy is found by applying Bellman optimality principles such that
	\begin{equation}
		{\bf \mu}^o_k=\argmin\limits_{{\bf \mu}_k} V({\bf Z}_{k})=-{\bf S}_{{\bf \mu} {\bf \mu}}^{-1}{\bf S}_{{\bf \mu} \mathcal{E}}  \, \mathcal{E}_k,\ k\in \mathbb{N}.
		\label{eq:OpPol}
	\end{equation}
	The simultaneous solution of \eqref{eq:Bello} and \eqref{eq:OpPol} provides a solution for the underlying Approximate Dynamic Programming (ADP) problem, namely the Action Dependent Heuristic Dynamic Programming (ADHDP). 
	The following definition is needed.
	\begin{defn}(Admissible Strategy~\cite{kyr1,morad2}) Let the set of all measurable maps ${\bf u}_k(.)\,:\, [a,b] \, \rightarrow U$ represent the space of admissible strategies where $U$ is a compact set. A control strategy $\omega$ is said to be admissible on $U$ if $V({\bf Z}_{0})$ is finite (i.e., $\lim\limits_{k \rightarrow \infty}\lVert{\bf \varepsilon}_k\rVert\rightarrow \bf \ell_2$).\hfill  $\square$
	\end{defn}
	
	The next Lemma shows that, solving the HJB \eqref{eq:HJB} or the Bellman optimality equation \eqref{eq:Bello} yields an asymptotically stable equilibrium point for the model-following error. 
	
	\begin{lem}
		\label{Lm:Lem1}
		Let the value function $V({\bf Z}_{0})$ and ${\bf y}^m_k, \forall k\in \mathbb{N}$ be bounded by $\bar \delta$ and $\bar \rho,$ respectively. Then, the equilibrium point of the model-following error system is asymptotically stable.
	\end{lem}
	\begin{proof} The value function $V(\cdot)$ is upper-bounded i.e., $V({\bf Z}_{k})\le \bar \delta$. Hence, the inequality $0 \le \dots \le V({\bf Z}_{k+1})\le V({\bf Z}_{k})\le \bar \delta$ holds for any admissible control strategy ${\bf \omega}$ and thus for the optimal strategy ${\bf \omega}^o$ as well. Hence, $V({\bf Z}_{k})-V({\bf Z}_{k+1}) \in \ell_\infty$ i.e., ${\bf \varepsilon}_k, \forall k \in \ell_\infty$ and ${\bf S} \in \ell_\infty$. Furthermore, the HJB equation \eqref{eq:HJB} results in $\left({\partial{V}_{k+1}}/{\partial {\bf Z}_{k+1}}\right)^\textrm{T} {\bf Z}_{k+1}\in \ell_\infty$ and consequently $\mathcal{E}_{k+1}\in \ell_\infty$.  Adopting an admissible strategy ${\bf \omega}$ or the optimal one yields a Bellman equation denoted by $V({\bf Z}_{k})=\frac{1}{2}{\bf \mathcal{E}}_k^\textrm{T} \, \left({\bf \mathcal Q} + {\bf \omega}^\textrm{T} \, {\bf \mathcal R}\,  {\bf \omega} \right)\, {\bf \mathcal{E}}_k + V({\bf Z}_{k+1})$. Therefore, the inequality $\frac{1}{2}{\bf \mathcal{E}}_k^\textrm{T} \, \left({\bf \mathcal Q} + {\bf \omega}^\textrm{T} \, {\bf \mathcal R}\,  {\bf \omega} \right)\, {\bf \mathcal{E}}_k \, \le V({\bf Z}_{k})$ holds. This implies that, ${\varepsilon_k} \in \ell_2$ and similarly $\left({\partial{\bf \lambda}_{k+1}}/{\partial {\bf Z}_{k+1}}\right)^\textrm{T} {\bf Z}_{k+1}\in \ell_2$. Hence, ${\varepsilon_k} \rightarrow 0$ as $k \rightarrow \infty$ and the resulting model-following error system has an asymptotically stable equilibrium point. 
	\end{proof}

	\section{PI Solution\label{sec:RL}}
	This section introduces a PI solution to solve the model-following control problem. The mechanism solves the Bellman optimality equation \eqref{eq:Bello} using the optimal strategy \eqref{eq:OpPol}. Further,the PI algorithm requires an initial admissible strategy before solving for better stabilizing strategies, as detailed in Algorithm~\ref{Alg:Alg1}. 
	\begin{algorithm}[htb!]
		\caption{\label{Alg:Alg1} Model-Following PI Solution}
		\begin{enumerate}
			\item[{\footnotesize{}1:}] Start with an admissible solution ${\bf S}^0$ and initialize the vector ${\bf X}_0$ (then calculate the error signal ${\bf \varepsilon}_0$ and vector ${\bf Z}_0$).
			\item[{\footnotesize{}2:}] Calculate the admissible control strategy ${\bf \omega}^0$.
			\item[{\footnotesize{}3:}] Evaluate the control strategy (i.e., obtain ${\bf S}^{j+1}$) by solving
			\begin{equation}
				V^{j+1}({\bf Z}^j_{k})-V^{j+1}({\bf Z}^j_{k+1})={\cal U}_k \, \left(\mathcal{E}_k,{\mu}^j_k\right),
				\label{eq:PIVal}
			\end{equation}
			where $j$ refers to a calculation step index.
			\item[{\footnotesize{}4:}] Improve the control strategy using
			\begin{equation}
				{\bf \mu}^{j+1}_k=-\, [{\bf S}_{{\bf \mu} {\bf \mu}}^{-1}{\bf S}_{{\bf \mu} \mathcal{E}}]^{j+1}  \, \mathcal{E}_k.
				\label{eq:PIPol}
			\end{equation}
			\item[{\footnotesize{}5:}] Terminate upon convergence of $\lVert{\bf S}^{j+1}-{\bf S}^{j}\rVert$.
		\end{enumerate}
	\end{algorithm}

	During the initial phase of the PI solution, vector ${\bf X}_0$ stores initial measurements of the dynamic system or process. Further, it requires an initial admissible strategy that can be chosen arbitrarily such that ${\bf S}^0={{\bf S}^0}^\textrm{T}\succ 0$ which is standard for PI \cite{Sutton_1998}.
	\begin{thm}
		\label{thm:Theorem2}
		Let the model-following solution be given by Algorithm~\ref{Alg:Alg1}, where the value function and the associated optimal strategy are given by $V\left({\bf Z}_{k}\right)$ and \eqref{eq:OpPol}, respectively. Then,
		\begin{enumerate}
			\item[a.] The strategies~\eqref{eq:PIPol} are stabilizing.
			\item[b.] The PI process yields a non-increasing positive definite solution sequence that follows $0\le \dots \le V\left(..\right)^o\le \dots \le \dots \le V\left(..\right)^1 \le V\left(..\right)^0$ with $V^o\left(..\right)$ being the optimal value function.
		\end{enumerate}
	\end{thm}
	
	\begin{proof}
		The proof follows from~\cite{AbouheafCTT2015} and \cite{Derong2014} and is omitted due to space limitations. 
	\end{proof}
	 This result affirms that, the simultaneous solution of~\eqref{eq:PIVal} and~\eqref{eq:PIPol} will lead to a sequence of stabilizing strategies, provided that an admissible strategy is adopted at the initialization phase of the RL solution. Further, the kernel solution matrix ${\bf S}$ will be non-increasing and bounded below.
	\begin{rem}
		\label{rem:rmk2}
		Lemma~\ref{Lm:Lem1} and Theorem \ref{thm:Theorem1} reveal that, if an initial admissible strategy is followed, and the model-following error system is stabilizable around the desired reference trajectory, then, Theorem \ref{thm:Theorem2} ensures a continuous strategy improvement along that desired trajectory. \hfill $\square$
	\end{rem}	
	
	\section{Actor-Critic Implementation} \label{sec:crit}
	We shall now derive tuning laws for the actor and critic approximators to solve Algorithm~\ref{Alg:Alg1} in an online fashion. The actor and critic structures approximate the strategy \eqref{eq:PIPol} and value function in~\eqref{eq:PIVal}, respectively. Then, Lagrange optimization principles are adopted to tune the weights in real-time.
	
	The critic approximates the value function $V({\bf Z}_{k})$ in~\eqref{eq:PIVal} such that 
	\begin{equation}
		\hat V({\bf Z}_{k})=\frac{1}{2}{\bf Z}^\textrm{T}_{k}\, {\bf \Theta}_{k}\, {\bf Z}_{k},
		\label{crit}
	\end{equation}
	where the matrix ${\bf \Theta}_{k}={\bf \Theta}^\textrm{T}_{k} \succ 0 \in\mathbb{R}^{((r + 1) p+t) \times ((r+ 1) p+t)}$ contains the critic weights.
		
		This value function can be reshaped such that  $\hat V({\bf Z}_{k})={\bf \bar \Theta}_{k}\, { \bf \bar Z}_{k}$, with ${\bf \bar Z}_{k}=\bigg\{\left({\bf Z}^\xi_{k} \bigotimes {\bf Z}^\zeta_{k}\right), \, \xi=1,\dots,((r+1) p+t),$ $\zeta=\xi,\dots,((r+1) p+t)\bigg\},$ $ { \bf \bar Z}_{k}\in \mathbb{R}^q, q= ((r + 1) p+t)  ((r + 1) p+t+1)/2,$ and ${\bf \bar \Theta}^\textrm{T}_{k} \in \mathbb{R}^q$ is a vector of the entries obtained from the matrix $\frac{1}{2}{\bf \Theta}_{k}$ that are associated with those of ${ \bf \bar Z}_{k}$. This form is more convenient to use in the policy improvement step of Algorithm~\ref{Alg:Alg1}.

	The optimal strategy \eqref{eq:OpPol} is then approximated using an actor adaptive structure such that
	\begin{equation}
		{\bf \hat \mu}(\mathcal{E}_{k})= {\bf \Omega}_{k}\,  {\bf \mathcal{E}}_{k}, \ \forall {\bf \mathcal{E}}_{k},
		\label{act}
	\end{equation}
	where ${\bf \Omega}_{k} \in\mathbb{R}^{t \times ((r+1) p)}$ are the actor weights of the strategy.
	
	This actor-critic structure can be easily employed by the temporal difference form \eqref{eq:PIVal} leading to a variety of RL solutions. Herein, the adaptation schemes of the actor-critic structures are inspired by Kaczmarz's projection approach~\cite{astrom2013}. 
	
	Let $\tilde V({\bf Z}_{k})={\bf \bar \Theta}_{k}$ with $ {\bf \tilde Z}_{k},$ ${\bf \tilde Z}_{k}= \left({ \bf \bar Z}_{k} \, - \,{ \bf \bar Z}_{k+1}\right)$ and $\tilde V^d\,\left({\bf Z}_{k}\right)={\cal U}_k \, \left(\mathcal{E}_k,{\bf \hat \mu}_k\right)$. It is required to choose the weights ${\bf \bar \Theta}_{k}$ to minimize $\lVert{\bf \bar \Theta}_k-{\bf \bar \Theta}_{k-1}\rVert$ subject to the constraint $\tilde V({\bf Z}_{k})=\tilde V^d({\bf Z}_{k})$. This is done using a Lagrange optimization process to minimize the function    
	$${\cal V}_{V_k}=\left({\bf \bar \Theta}_k-{\bf \bar \Theta}_{k-1}\right)\left({\bf \bar \Theta}_k-{\bf \bar \Theta}_{k-1}\right)^\textrm{T}+\alpha_V \left(\tilde V({\bf Z}_{k})-\tilde V^d({\bf Z}_{k})\right),$$ 
	where $\alpha_V$ is a Lagrange multiplier.
	Applying the optimization principles (i.e., $\displaystyle {\partial {\cal V}_{V_k}}/{\partial {\bf \bar \Theta}_k}=0$ and $\displaystyle   {\partial {\cal V}_{V_k}}/{\partial \alpha_V}=0$) yields 
	\begin{eqnarray*}
		\left({\bf \bar \Theta}_k-{\bf \bar \Theta}_{k-1}\right)^\textrm{T}+\alpha_V \, {\bf \tilde Z}_{k}=0 \text{ and }
		{\bf \bar \Theta}_{k} \, {\bf \tilde Z}_{k} \,-\,\tilde V^d \left({\bf Z}_{k}\right)=0.
	\end{eqnarray*}
	Hence, $\left({\bf \bar \Theta}_k-{\bf \bar \Theta}_{k-1}\right)\,{\bf \tilde Z}_{k}+\alpha_V \, {\bf \tilde Z}^\textrm{T}_{k}\,{\bf \tilde Z}_{k}=0$ and $\displaystyle \alpha_V=\frac{-1}{{\bf \tilde Z}^\textrm{T}_{k} \,{\bf \tilde Z}_{k}}\left(\tilde V^d\left({\bf Z}_{k}\right)-{\bf \bar \Theta}_{k-1}{\bf \tilde Z}_{k}\right)$. Therefore, 
	$
	{\bf \bar \Theta}_k\,=\,{\bf \bar \Theta}_{k-1}\,-\,\alpha_V \, {\bf \tilde Z}^\textrm{T}_{k} 
	$
	and
	$\displaystyle 	{\bf \bar \Theta}_k={\bf \bar \Theta}_{k-1}- \,  
	\frac{{\bf \tilde Z}^\textrm{T}_{k}}{{\bf \tilde Z}^\textrm{T}_{k}\,{\bf \tilde Z}_{k}}\left({\bf \bar \Theta}_{k-1}{\bf \tilde Z}_{k}\,-\,\tilde V^d\left({\bf Z}_{k}\right)\right).$
	%
	To manage the update steps of the adapted critic weights, a factor $\delta_V$ is considered in addition to a constant $\eta_V$ to avoid the singularity issue when updating the critic weights to write
	\begin{equation}
		{\bf \bar \Theta}_k={\bf \bar \Theta}_{k-1}- \,  
		\frac{\delta_V \, {\bf \tilde Z}^\textrm{T}_{k}}{\eta_V+{\bf \tilde Z}^\textrm{T}_{k}{\bf \tilde Z}_{k}}\left({\bf \bar \Theta}_{k-1}{\bf \tilde Z}_{k}-\tilde V^d({\bf Z}_{k})\right).
		\label{eq:crtup}
	\end{equation}
	The weights ${\bf \Theta}_k$ can be reconstructed from the solution ${\bf \bar \Theta}_k$. 
	In a similar fashion, the actor weights are selected to minimize the approximation error $\lVert{\bf \Omega}_k-{\bf \Omega}_{k-1}\rVert$ subject to the constraint ${\bf \hat \mu}(\mathcal{E}_{k})={\bf \hat \mu}^d(\mathcal{E}_{k})$ where ${\bf \hat \mu}^d(\mathcal{E}_{k})=-{\bf \Theta}_{{\bf \mu} {\bf \mu}}^{-1}{\bf \Theta}_{{\bf \mu} \mathcal{E}}  \, \mathcal{E}_k$. Hence, the adaptation error of the actor weights, representing the optimal strategy, can be minimized using the following function 
	$${\cal V}_{\mu_k}=\left({\bf \Omega}_k-{\bf \Omega}_{k-1}\right)\left({\bf \Omega}_k-{\bf \Omega}_{k-1}\right)^\textrm{T}+\alpha_\mu \left({\bf \hat \mu}(\mathcal{E}_{k})-{\bf \hat \mu}^d(\mathcal{E}_{k})\right),$$
	where $\alpha_\mu$ is a Lagrange multiplier. Applying the optimality conditions (i.e., $\displaystyle {\partial {\cal V}_{\mu_k}}/{\partial {\bf \Omega}_k}=0$ and $\displaystyle {\partial {\cal V}_{\mu_k}}/{\partial \alpha_\mu}=0$) yields 
	\begin{eqnarray*}
		\left({\bf \Omega}_k-{\bf \Omega}_{k-1}\right)^\textrm{T}+\alpha_\mu \, \mathcal{E}_{k}=0 \text{ and }
		{\bf \Omega}_{k}\, \mathcal{E}_{k}-{\bf \hat \mu}^d(\mathcal{E}_{k})=0.
	\end{eqnarray*}
	Hence, 
	$
	\left({\bf \Omega}_k-{\bf \Omega}_{k-1}\right)\, \mathcal{E}_{k}+\alpha_\mu \, \mathcal{E}^\textrm{T}_{k}\mathcal{E}_{k}=0
	$
	and $ \displaystyle \alpha_\mu=\frac{-1}{\mathcal{E}^\textrm{T}_{k}\mathcal{E}_{k}}\left({\bf\hat \mu}^d(\mathcal{E}_{k})-{\bf \Omega}_{k-1}{\bf \mathcal{E}}_{k}\right).
	$
	The actor tuning law is given by
	$\displaystyle 
	{\bf \Omega}_k={\bf \Omega}_{k-1}- \,  
	\frac{\mathcal{E}^\textrm{T}_{k}}{\mathcal{E}^\textrm{T}_{k}\mathcal{E}_{k}}\left({\bf \Omega}_{k-1}{\bf \mathcal{E}}_{k}-{\bf\hat \mu}^d(\mathcal{E}_{k})\right).
	$
	Similarly, a refined actor adaptation law is given by		\begin{equation}
		{\bf \Omega}_k={\bf \Omega}_{k-1}- \,  
		\frac{\delta_\mu \,  \mathcal{E}^\textrm{T}_{k}}{\eta_\mu+\mathcal{E}^\textrm{T}_{k}\mathcal{E}_{k}}\left({\bf \Omega}_{k-1}{\bf \mathcal{E}}_{k}-{\bf \hat \mu}^d(\mathcal{E}_{k})\right),
		\label{eq:actup}
	\end{equation}
	where $\delta_\mu$ and $\eta_\mu$ are constants to control the update steps of the adapted actor weights and to avoid singularity when $\mathcal{E}_{k}=0$, respectively.

	The next result shows how to get approximate bounds on $\delta_V$ and $\delta_\mu$ to ensure convergence of the adapted weights.
	\begin{lem}
		\label{lem:const}
		Let ${\bf \Theta}^{o}$ and ${\bf \Omega}^{o}$ be the optimal weights representing the solution of the Bellman equation \eqref{eq:Bello} and the optimal control gains \eqref{eq:OpPol}, respectively. Then, choosing $0<\delta_V<2$ and $0<\delta_\mu<2$ yields bounded deviations of the critic and actor weights from the optimal weights.  
	\end{lem}
	\begin{proof}
		According to \eqref{eq:crtup}, the temporal difference errors in the updated critic weights are given by 
		$ \displaystyle 
		{\bf \bar \Theta}^{e}_k={\bf \bar \Theta}^e_{k-1}- \,  
		\frac{\delta_V \, \left({\bf \bar \Theta}^e_{k-1}\,{\bf \tilde Z}_{k}\right){\bf \tilde Z}^\textrm{T}_{k}}{\eta_V+{\bf \tilde Z}^\textrm{T}_{k}\,{\bf \tilde Z}_{k}},
		$
		where ${\bf \bar \Theta}^e_k\,=\,{\bf \bar \Theta}^o_k\,-\,{\bf \bar \Theta}_k$ and ${\bf \bar \Theta}^e_{k-1}\,=\,{\bf \bar \Theta}_{k-1}{\bf \tilde Z}_{k}\,-\,\tilde V^d\left({\bf Z}_{k}\right)$.
		Then $\displaystyle {\bf \bar \Theta}^{e \, T}_k=\left({\bf I}-\frac{\delta_V \, {\bf \tilde Z}_{k} \, {\bf \tilde Z}^\textrm{T}_{k}}{\eta_V+{\bf \tilde Z}^\textrm{T}_{k} \, {\bf \tilde Z}_{k}}\right)\,{\bf \bar \Theta}^{e\, T}_{k-1},$ where ${\bf I} \in \mathbb{R}^{((r + 1) p+t) \times ((r+ 1) p+t)}$ is an identity matrix. This dynamical expression has an eigenvalue $\displaystyle \Lambda_V=\left(\frac{\eta_V +(1-\delta_V) \, {\bf \tilde Z}^\textrm{T}_{k}{\bf \tilde Z}_{k}}{\eta_V+{\bf \tilde Z}^\textrm{T}_{k}{\bf \tilde Z}_{k}}\right)$ that is less than 1 if $0<\delta_V<2$ while the remaining eigenvalues are equal to 1. Similarly according to \eqref{eq:actup}, the temporal difference errors in the updated actor weights are given by $\displaystyle {\bf \Omega}^{e\, T}_k=\left({\bf \mathcal I}-\frac{\delta_\mu \, \mathcal{E}_{k}\mathcal{E}^\textrm{T}_{k}}{\eta_\mu+\mathcal{E}^\textrm{T}_{k}\mathcal{E}_{k} }\right)\,{\bf \Omega}^{e\, T}_{k-1},$ where ${\bf \Omega}^e_{k}={\bf \Omega}^o_{k}-{\bf \Omega}_{k}$ and ${\bf \mathcal I} \in \mathbb{R}^{((r + 1) p) \times ((r+ 1) p)}$ is an identity matrix. Accordingly, one eigenvalue associated with this form $\displaystyle\Lambda_\mu=\left(\frac{\eta_\mu+(1-\delta_\mu) \, \mathcal{E}^\textrm{T}_{k}\mathcal{E}_{k}}{\eta_\mu+\mathcal{E}^\textrm{T}_{k}\mathcal{E}_{k} }\right)$ is less than 1 when $0<\delta_\mu<2$. These results are true for $\eta_V>0$ and $\eta_\mu>0$.
	\end{proof}

\begin{algorithm}[htb!]
	\setstretch{1} 
	\caption{\label{Alg:act-crt}Implementation of the Actor-Critic Solution}
	\begin{algorithmic}[1] 
		\Require
		\Statex Number of calculation steps ${\cal N}_T$.
		\Statex Constants $\delta_V, \delta_\mu, \eta_V,$ and $\eta_\mu$.
		\Statex Initial actor ${\bf \Omega}_{0}$ and critic ${\bf \Theta}_{0}$ weights.
		\Statex Performance index weighting matrices $\boldsymbol{\mathcal Q}$ and $\bm{\mathcal  R}$.
		\Statex Convergence threshold ${\mathcal T_r}$ observed within a finite-horizon of ${\mathcal N}$ calculation steps.
		\Ensure
		\Statex Converged actor-critic weights ${\bf \Omega}^{o}$ and ${\bf \Theta}^{o}$.
		\Statex
		\State Initialize ${\bf X}_{0},$ $\mathcal{E}_{0},$ and ${\bf \Theta}_{0}$ \Comment{Use admissible strategy ${\bf \Omega}_{0}$}.
		\State Introduce the reference signal ${\bf y}^m_0$.
		\State $k \gets 0$ 
		\State Convergence-of-Actor-Critic-Weights $\gets$ False
		\While {Convergence-of-Actor-Critic-Weights $=$ False \AlgAnd $k \le {\cal N}_T$ }
		\State Calculate $\tilde V^d({\bf Z}_{k})={\cal U}_k \, \left(\mathcal{E}_k,{\bf \hat \mu}_k\right)$.
		\State Apply ${\bm \hat \mu}_{k}$ to the process (i.e., \eqref{eq:dyn}) and get ${\bm y}_{k+1}$ .
		\State Obtain the error vector $\mathcal{E}_{k+1}$ and then find an estimate for the control signal ${\bm \hat \mu}_{k+1}$ using \eqref{act}.
		\State Find $\hat V({\bf Z}_{k})$ and $\hat V({\bf Z}_{k+1})$ in order to get $\tilde V({\bf Z}_{k})$.  
		\State Adapt the critic weights ${\bf \Theta}_{0}$ following \eqref{eq:crtup}. 
		\State Adapt the actor weights ${\bf \Omega}_{0}$ following \eqref{eq:actup}. 
		\If{$\lVert{\bf \Theta}_k-{\bf \Theta}_{k-1}\rVert$\AlgAnd $\lVert{\bf \Omega}_k-{\bf \Omega}_{k-1}\rVert$ converge} 
		\State ${\bm\Theta}^{(o)} \gets {\bm\Theta}_{(k+1)}$ and ${\bm\Omega}^{(o)} \gets {\bm\Omega}_{(k+1)}$
		\State Convergence-of-Actor-Critic-Weights $\gets$ True 
		\EndIf
		\State $k \leftarrow k+1$
		\EndWhile
		\Statex \Return adapted weights ${\bf \Theta}_k$ and ${\bf \Omega}_k$, for $k=0,1,\dots, \cal{N}_T$
	\end{algorithmic}
\end{algorithm}
	
	\begin{rem}
		\label{rem:rmk3}
		Herein, the structure of the actor is chosen to be linear to enable the adopted optimal control setup and the associated temporal difference solution. This is convenient to many RL computational setups. It may not capture ultimately the behavior of highly nonlinear complicated systems, where other nonlinear forms of neural networks may be convenient. The steps of the online actor-critic solution are shown in Algorithm~\ref{Alg:act-crt} where probing noise is used for appropriate state exploration (persistence of excitation).
		 The initial phase of the RL solution aims to select arbitrarily the critic weights (i.e., ${\bf \Theta}_{0}={{\bf \Theta}_{0}}^\textrm{T}\succ 0$) to obtain an admissible strategy ${\bf \Omega}_{0}$, as emphasized in Theorem~\ref{thm:Theorem2}. Further, the constants $\delta_V$ and $\delta_\mu$ must be decided according to the conditions designated by Lemma~\ref{lem:const} (i.e., $0<\delta_V<2$ and $0<\delta_\mu<2$). Finally, the values of $\boldsymbol{\mathcal Q}$ and $\bm{\mathcal  R}$ are selected to achieve the intended optimization objectives, namely the paces at which the tracking error dynamics $\mathcal{E}_k$ and the correction control signal ${\bf \mu}^{\bf \omega}_k$ are regulated.	
		The adaptive solution developed herein relies on solving two optimization problems. The first one provides a PI solution to a Bellman's optimality equation \eqref{eq:Bello} following an optimal strategy \eqref{eq:OpPol}. This solution generates a sequence of non-increasing value functions $0 < V\left(..\right)^o\le \dots \le V\left(..\right)^1 \le V\left(..\right)^0$. The second problem optimizes the performance of the actor and critic adaption schemes by projecting the vectors ${\bf \Omega}_{k}$ and ${\bf \bar \Theta}_{k-1}$ on the vectors $\mathcal{E}_{k}$ and ${\bf \tilde Z}_{k}$, respectively.  \hfill  $\square$
	\end{rem}		
	
	\section{Simulation Results\label{sec:Sim}}
	The efficacy of the model-free adaptive learning solution is tested using: (i) a linear system with state and input delays and (ii) a nonlinear system. Further, two model-following approaches based on sliding mode and high-order model-free adaptive control schemes are considered for comparison purposes~\cite{MPC2014,MFAC2021}. 
	
	\subsection{Case 1: Linear System with State and Input Delays}
	An Autonomous Underwater Vehicle (AUV) is adopted to validate the online model-free projection solution~\cite{AUV1,MPC2014}. A linear system with state and input delays is given by ${\bf X}_{k+1}={\bf A} \, {\bf X}_{k}+{\bf B} \, {\bf u}_{k}+{\bf A}_d \, {\bf X}_{k-d}+{\bf B}_h\, {\bf u}_{k-h},{\bf y}_{k}={\bf C} \, {\bf X}_{k},$ where $\footnotesize{\bf A}:=\left[
	\begin{array}{ccc}
		0.9817 &-0.0119 &0\\
		0.0099  &0.9999 &0\\
		0      &-0.01  & 1
	\end{array}
	\right],$ 
	$\footnotesize{\bf A}_d:=\left[
	\begin{array}{ccc}
		0.0099 & 0.005 &  0.005\\
		0     & -0.001 & -0.0005\\
		-0.001 & -0.0005 & 0.001
	\end{array}
	\right],{\bf B}:=\left[
	\begin{array}{c}
		-0.0131\\
		-0.0001\\
		0
	\end{array}
	\right],$ $\footnotesize {\bf B}_h:=\footnotesize \left[\begin{array}{c}
		0.001 \\  0.0001 \\ 0.0001
	\end{array}
	\right],
	{\bf X}_0 =\footnotesize \left[ \begin{array}{c}
		-0.2956\\
		-0.7210\\
		-1.7932
	\end{array}
	\right],
	$
	$\footnotesize {\bf C}:=\left[
	\begin{array}{ccc}
		-0.8784 &  -2.3961  &  0.6464
	\end{array}
	\right], d:=10,$ and $h:=20$. The AUV follows a dynamic trajectory given by ${\bf X}^m_{k+1}={\bf A}^m \, {\bf X}^m_{k}$ and ${\bf y}^m_{k}={\bf C}^m \, {\bf X}^m_{k},$ where  $\footnotesize{\bf A}^m:=\left[
	\begin{array}{ccc}
		1&0.01 &  \\
		-0.01 & 1 
	\end{array}
	\right]$ and $\footnotesize{\bf C}^m:=\left[
	\begin{array}{ccc}
		1 & 0
	\end{array}
	\right]
	$. 
	These dynamic forms are used to observe the model-following errors. However, neither the dynamics of the AUV nor those of the model-to-follow are employed explicitly to find the online solution. The  discrete-time AUV dynamic system is sampled from a continuous-time dynamic model with a sampling time of $T_s=0.01$ second~\cite{AUV1,MPC2014}. It is noted that, using three error samples (i.e., $r=2$) is found to be sufficient to pick the model-following error dynamics. The remaining simulation parameters are listed in Table~\ref{Tab1}. 
	
	The PI mechanism employs a probing noise during the first $2.5$ seconds to satisfy the persistence of excitation condition. This propping noise is useful to better explore the dynamic environment without completely exploiting the strategies. The RL solution is compared with a robust model-following approach that is based on Sliding Mode Control (SMC)~\cite{MPC2014}. This approach relies on knowing the full dynamical information of the AUV and the model-to-follow as well. The steps of the MPC solution can be summarized as follows~\cite{MPC2014}. 
	\begin{enumerate}
		\item Find $\bf G$ and $\bf H$ to satisfy $\footnotesize\left[
		\begin{array}{cc}
			\bf{A}&\bf{B}   \\
			\bf{C} & {0} 
		\end{array}
		\right]\left[
		\begin{array}{c}
			\bf{G}   \\
			\bf{H} 
		\end{array}
		\right]=\left[
		\begin{array}{c}
			\bf{G} \,  \bf{A}^m   \\
			\bf{C}^m 
		\end{array}
		\right]$. According to (\cite{GH90}), $\footnotesize\bf{G}=\left[
		\begin{array}{cc}
			-0.1785 &   0.0292\\
			-0.3275  & -0.0152\\
			0.0905  & -0.0326\\
		\end{array}
		\right], {\bf H}=\left[
		\begin{array}{cc}
			0.9999 &  -0.0100
			
		\end{array}
		\right]$.
		\item ${\bf\tau}_k={\bf X}_{k}-{\bf G}{\bf X}^m_{k}$.
		\item ${E}_{k}={E}_{k-1}+{\bf S}_f\,{\bf (A+B\, K)}{\bf\tau}_{k-1}+{\bf S}_f\,{\bf A}_d{\bf\tau}_{k-11}+{\bf S}_f\,{\bf B}_h{\bf\tau}_{k-21}$, ${E}_{0}=0$, where $\bf{S}_f=\left[
		\begin{array}{ccc}
			1 &1.85& -0.825
		\end{array}
		\right]$ and $\bf{K}=\left[
		\begin{array}{ccc}
			145.9573 &270.1303& -120.8601
		\end{array}
		\right]$.
		\item $\sigma_k={\bf S}_f\,{\bf\tau}_k-{\bf S}_f\,\textrm{exp}(-0.1\, k){\bf\tau}_0-{E}_{k}$.
		\item ${R}_k={R}_{k-1}+\sigma_k-0.5\,\delta_{k-1}$.
		\item The control signal is calculated using a form given by $\footnotesize {\bf u}_{k}={\bf H} \, {\bf X}^m_{k}+\left[
		\begin{array}{ccc}
			145.9573 & 270.1303 &-120.8601
		\end{array}
		\right]\,{\bf\tau}_k+\left[
		\begin{array}{ccc}
			-75.6846 & -140.0165 & 62.4398
		\end{array}
		\right]\textrm{exp}(-0.1(k+1)){\bf\tau}_0+75.6846 \, R_k-75.6846\, E_k-37.8423 \,\sigma_k$.
	\end{enumerate}
	
	\begin{table*}
		\caption{Simulation and Learning Parameters}
		\label{Tab1}
		\begin{tabularx}{\textwidth}{@{}l*{11}{C}c@{}}
			\toprule
			Parameter   &  Value & Parameter   &  Value & Parameter   &  Value & Parameter   &  Value & Parameter   &  Value\\
			\midrule
			${\bf \mathcal Q}$ & $0.05\,I_{3}$   & ${\bf \mathcal R}$   & $0.01\,I_{3}$ & ${\cal N}_T$  & $4000$ & ${\mathcal N}$ & $30$&
			$r$ & $2$ \\ ${\mathcal T_r}$ &	$0.0005$	& 
			$\delta_V$    & $0.5$  & $\eta_V$   & $1.5$&
			$\delta_\mu$  & $0.5$  & $\eta_\mu$   & $1.5$ &
			\\
			\bottomrule	\end{tabularx}
	\end{table*}

	Fig. \ref{fig:fig1} presents the simulation results of the online RL and SMC solutions. The tuned actor-critic weights are demonstrated to converge after some exploration phase, as illustrated by Fig.~\ref{fig:act1}~and~\ref{fig:crt1}. The resulting control signal follows the form depicted by Fig. \ref{fig:con1}. The RL solution shows appropriate model-following after $4$ seconds, while the SMC will exhibit a small model-following error offset of $0.1$ as highlighted by Fig.~\ref{fig:ref1}~and~\ref{fig:err1}. Unlike the SMC solution, the RL solution approach does not employ any explicit knowledge of the model-to-follow and the AUV dynamics to calculate the model-following control strategy. 
	%
	The actor weights converge during the first $7$ seconds and the resulting strategy has gains given by
	$ {\bf \Omega}=
	\left[
	\begin{array}{ccc}
		4.0168 &  -0.2670 &  -2.7488
	\end{array}
	\right].$
	
	\subsection{Case 2: Nonlinear System}
	In this case, our RL solution will be compared to an improved high-order Model Free Adaptive Control (MFAC) approach~\cite{MFAC2021}. This is simulated using a nonlinear dynamical process described by
	\begin{align*}
		{\bf y}_{k+1}
		& \small
		=
		\begin{cases}
			\small
			\frac{{\bf y}_{k}}{1+{\bf y}^2_{k}}+{\bf u}^3_{k}, \hspace*{1em} \text{for }  k \le {\cal N}_T/2\\
			\displaystyle
			\frac{{\bf y}_{k} \, {\bf y}_{k-1} \, {\bf y}_{k-2} \, {\bf u}_{k-1}\, \left({\bf y}_{k-2}-1\right)+{\textrm{round}(2\,k/{\cal N}_T)}{\bf u}_{k}}{1+{\bf y}^2_{k-1}+{\bf y}^2_{k-2}}, \\ \hspace*{1em} \text{for } {\cal N}_T/2 < k \le {\cal N}_T
		\end{cases}
	\end{align*}
	where $\textrm{round}(\cdot)$ returns the value of a number rounded to the nearest integer.
	The desired response follows a dynamic behavior given by
	\begin{align*}
		\footnotesize
		{\bf y}^m_{k+1}
		&\footnotesize
		=
		\begin{cases}
			\footnotesize
			0.5\,\sin\left(k\,\pi/100\right)+0.3\,\cos\left(k\,\pi/50\right),  \text{for }  k \le {\cal N}_T/5\\
			0.5 \, (-1)^{\textrm{round}(2\,k/{\cal N}_T)},  \text{for } {\cal N}_T/5 < k \le 2{\cal N}_T/5\\
			0.5\,\sin\left(k\,\pi/100\right)+0.3\,\cos\left(k\,\pi/50\right),  \text{for } 2{\cal N}_T/5 < k \le 4{\cal N}_T/5\\
			-0.4 \, (-1)^{\textrm{round}(2\,k/{\cal N}_T)}, \text{for } 4{\cal N}_T/5 < k \le {\cal N}_T.
		\end{cases}
	\end{align*}

	The improved high-order MFAC approach is implemented using a set of dynamic recursive equations as follows~\cite{MFAC2021}; First, an estimation law is calculated such that
	\begin{eqnarray*}
		\tiny
		\phi_k&=&\sum_{i=1}^{6}\beta(i)\phi_{k-i}+\frac{0.8({\bf u}_{k-1}-{\bf u}_{k-2})}{0.01+({\bf u}_{k-1}-{\bf u}_{k-2})^2}\left(({\bf y}_{k}-{\bf y}_{k-1})-\right.\\ 
		& &\tiny \left. ({\bf u}_{k-1}-{\bf u}_{k-2})\sum_{i=1}^{6}\beta(i)\phi_{k-i}\right), \phi_0=0.5.
	\end{eqnarray*}
	Then, the control law is computed as follows
	\begin{eqnarray*}
		\tiny
		{\bf u}_{k}&=& \frac{\phi^2_k}{0.1+\phi^2_k}{\bf u}_{k-1}+ \frac{0.1}{0.1+\phi^2_k}\sum_{i=1}^{4} \bar\beta(i){\bf u}_{k-i}+\frac{0.8 \,\phi_k\left({\bf y}^m_{k}-{\bf y}_{k}\right)}{0.1+\phi^2_k},
	\end{eqnarray*}
	where $\small\beta=\left[
	\begin{array}{cccccc}
		1/2 &1/4 &1/8 &1/16& 1/32 & 1/32
	\end{array}
	\right]$ and $\small \bar \beta=\left[
	\begin{array}{cccccc}
		1/2 &1/4 &1/8 &1/8
	\end{array}
	\right]$.

	%
	%
	The simulation results are shown in Fig. \ref{fig:fig2}. The actor-critic weights are shown to converge after some initial learning phase as demonstrated by Fig.~\ref{fig:act2}~and~\ref{fig:crt2}.
	The RL solution is shown to outperform the improved high-order MFAC as depicted from the analysis ofthe model-following and error performances (see Fig.~\ref{fig:ref2}~and~\ref{fig:err2}). Further,the RL solution can capture the abrupt changes in the dynamics of the reference model. Fig. \ref{fig:con2} shows the control signal, where the gains of the control strategy converge to  $ {\bf \Omega}=
	\left[
	\begin{array}{ccc}
		0.9786  & -0.1755  &  0.2834
	\end{array}
	\right]
	$. Therefore, applying probing noise, considering nonlinear and linear systems with delays, and employing drastic nonlinear model-reference forms challenged the robustness of the RL solution.
	
	\subsection{Discussion}	
	The simulation outcomes consolidated the theoretical setup in terms of the stability and convergence aspects. First, Assumptions~\ref{Asm:1}~and~\ref{Asm:2} leading to Theorem~\ref{thm:Theorem1} and Lemma~\ref{Lm:Lem1}, revealed the asymptotic stability characteristics of the RL solution. The tracking error dynamics are shown to be stabilized asymptotically when simulated for dynamic model-following scenarios and using systems of nonlinear and linear types with delays, as demonstrated by Fig.~\ref{fig:err1}~and~\ref{fig:err2}. Second, Theorem~\ref{thm:Theorem2} guarantees an improved sequence of stabilizing policies which is revealed by Fig.~\ref{fig:act1}~and~\ref{fig:act2}. Furthermore, the convergence of the adapted actor and critic weights is supported by Lemma~\ref{lem:const}.
	
	 This work presents an adaptive strategy ${\bf \Omega}$ which can be easily i) adopted in a digital setup without using complex function approximators (i.e., employing \eqref{eq:PIVal}, \eqref{eq:PIPol}, \eqref{eq:crtup}, and \eqref{eq:actup}), ii) implemented in a model-free and data-driven fashion (i.e., using \eqref{eq:PIPol} and \eqref{eq:actup}), iii) configurable to a desired order of error dynamics (see $\mathcal{E}_k$ and Fig. \ref{fig:fig}), and iv) adapted to reflect the attainment of simultaneous optimization goals (see Theorem~\ref{thm:Theorem1} and Lemmas~\ref{Lm:Lem1} and ~\ref{lem:const}). 
	

	\begin{figure*}[!hbt]
		\centering
		\subcaptionbox{The evolution of the actor weights~${\bf \Omega}_{k}$.%
			\label{fig:act1}}
		{%
			\includegraphics[width=0.33\textwidth]{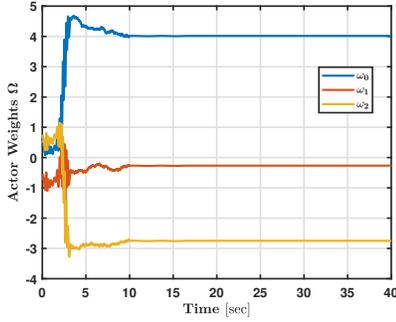}%
		}
		\hfill \hspace{-15pt}
		\subcaptionbox{The evolution of the critic weights~${\bf \bar \Theta}_{k}$.%
			\label{fig:crt1}}
		{%
			\includegraphics[width=0.33\textwidth]{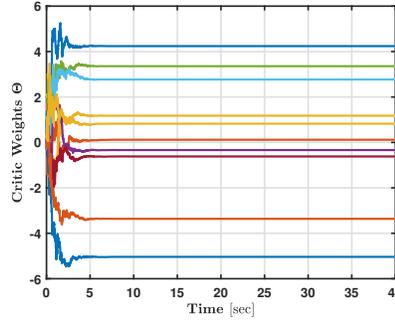}%
		}
		\subcaptionbox{The evolution of the control signal~${\bf u}_k$.%
			\label{fig:con1}}
		{%
			\includegraphics[width=0.33\textwidth]{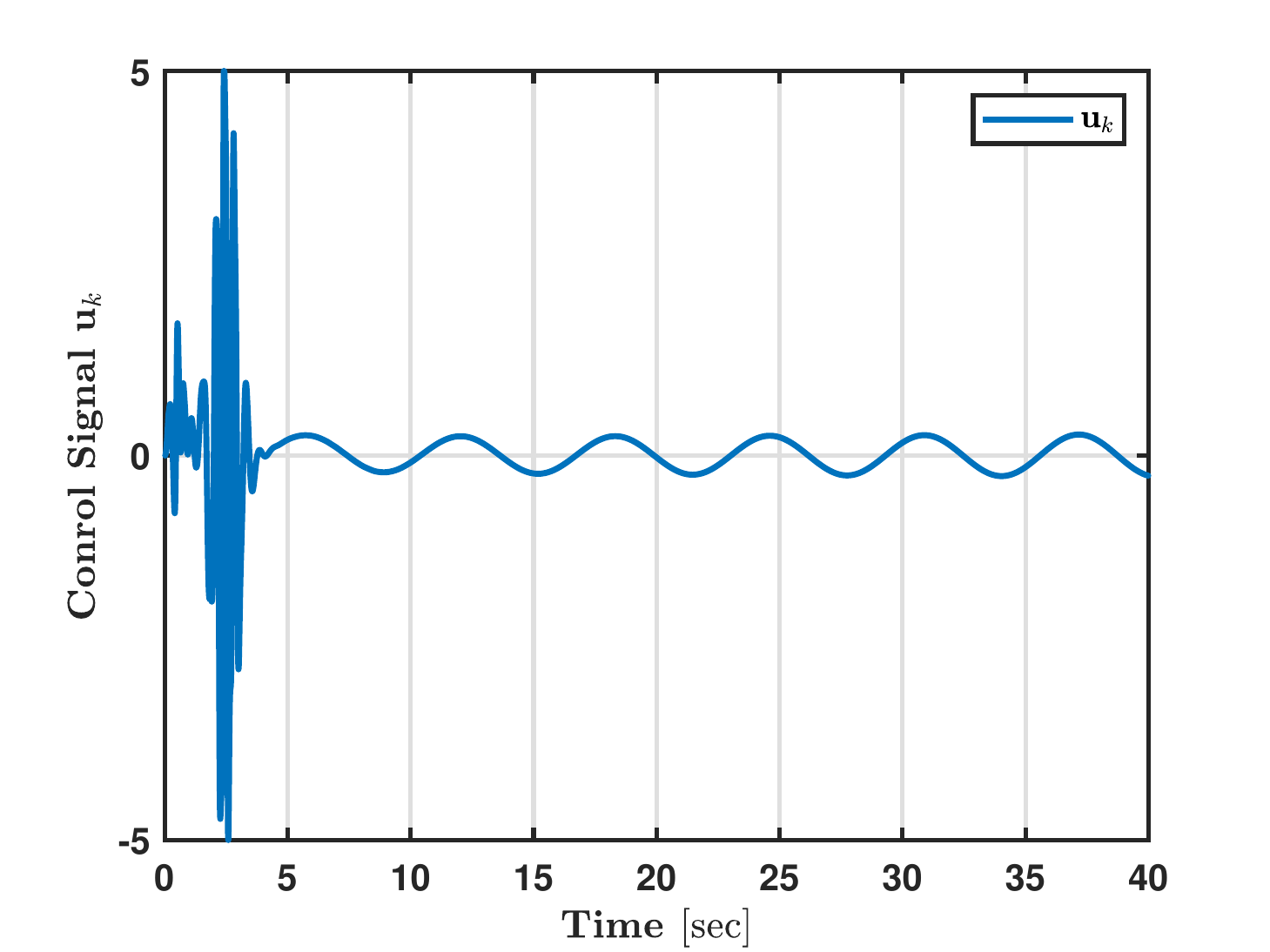}%
		}
		\hfill \hspace{-15pt}
		\\[2ex]
		\mbox{}\hfill
		\subcaptionbox{The evolution of the trajectory-tracking performance.%
			\label{fig:ref1}}
		{%
			\includegraphics[width=0.35\textwidth]{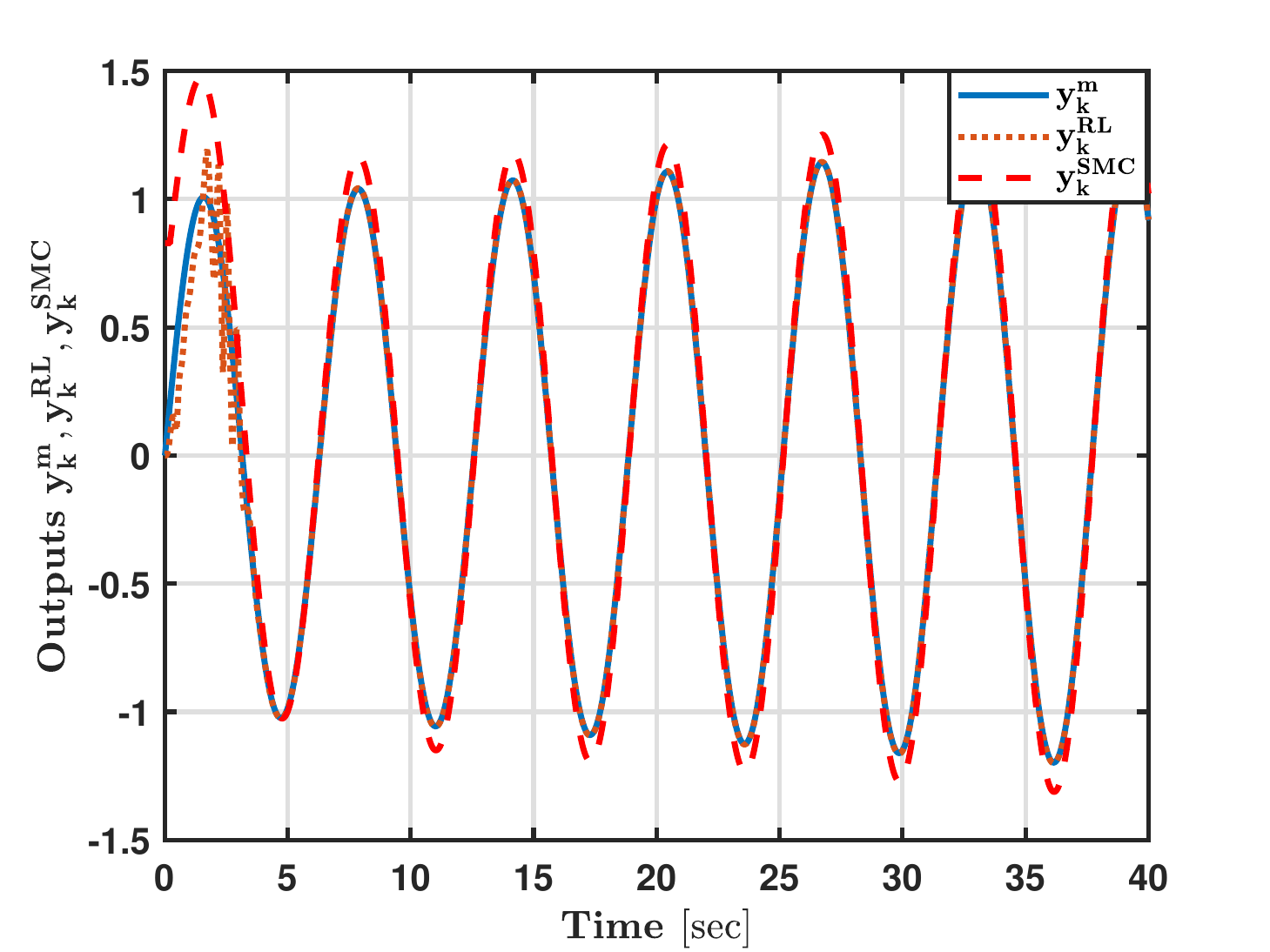}%
		}
		\hfill
		\subcaptionbox{The evolution of the error signal ${\bf \varepsilon}_k$.%
			\label{fig:err1}}
		{%
			\includegraphics[width=0.35\textwidth]{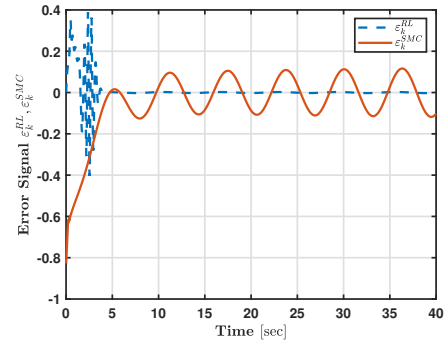}%
		}
		\hfill\mbox{}
		\caption{Simulation results of Case~1.\label{fig:fig1}} 
	\end{figure*}
	
	\begin{figure*}[!hbt]
		\centering
		\subcaptionbox{The evolution of the actor weights~${\bf \Omega}_{k}$.%
			\label{fig:act2}}
		{%
			\includegraphics[width=0.33\textwidth]{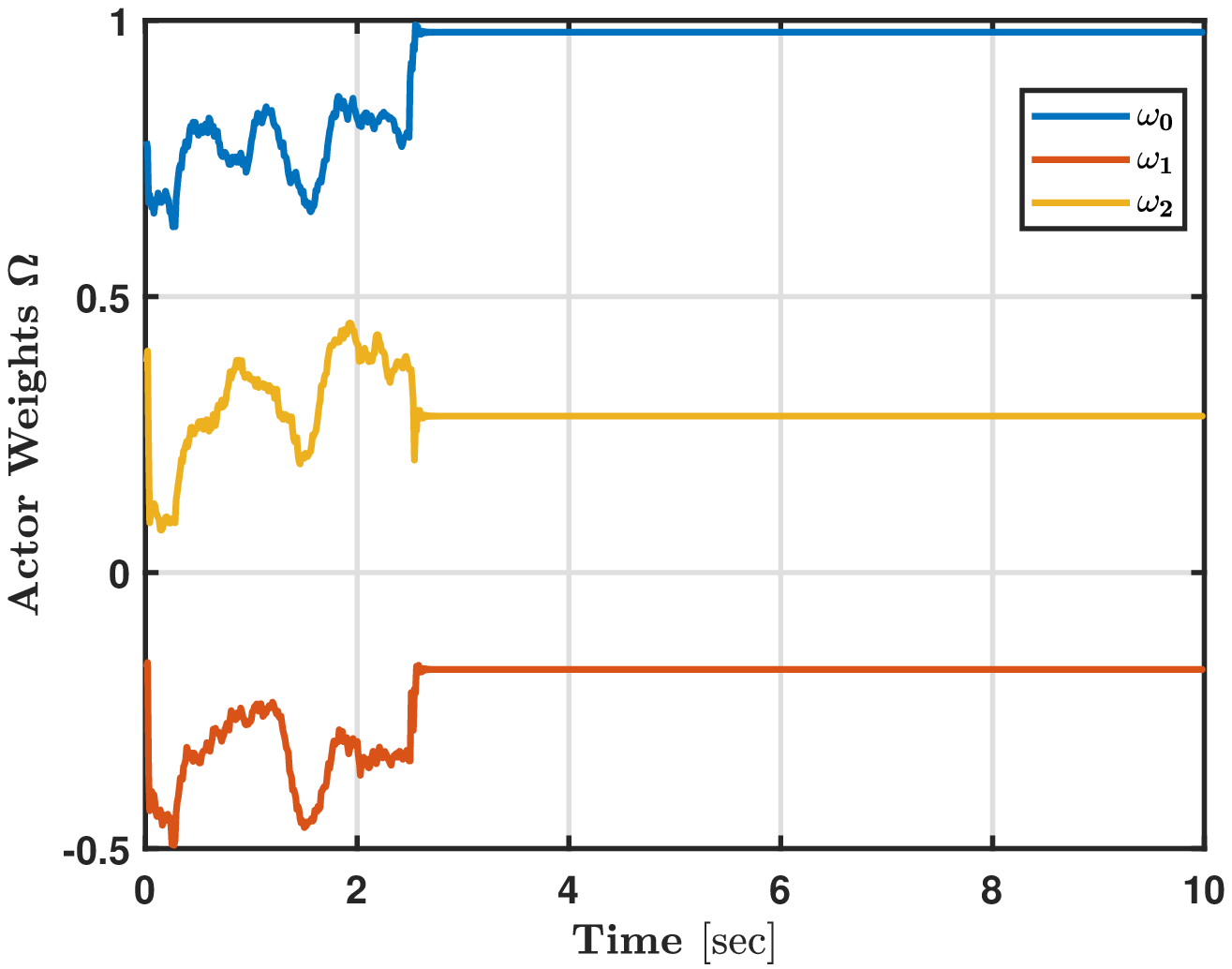}%
		}
		\hfill \hspace{-15pt}
		\subcaptionbox{The evolution of the critic weights~${\bf \bar \Theta}_{k}$.%
			\label{fig:crt2}}
		{%
			\includegraphics[width=0.33\textwidth]{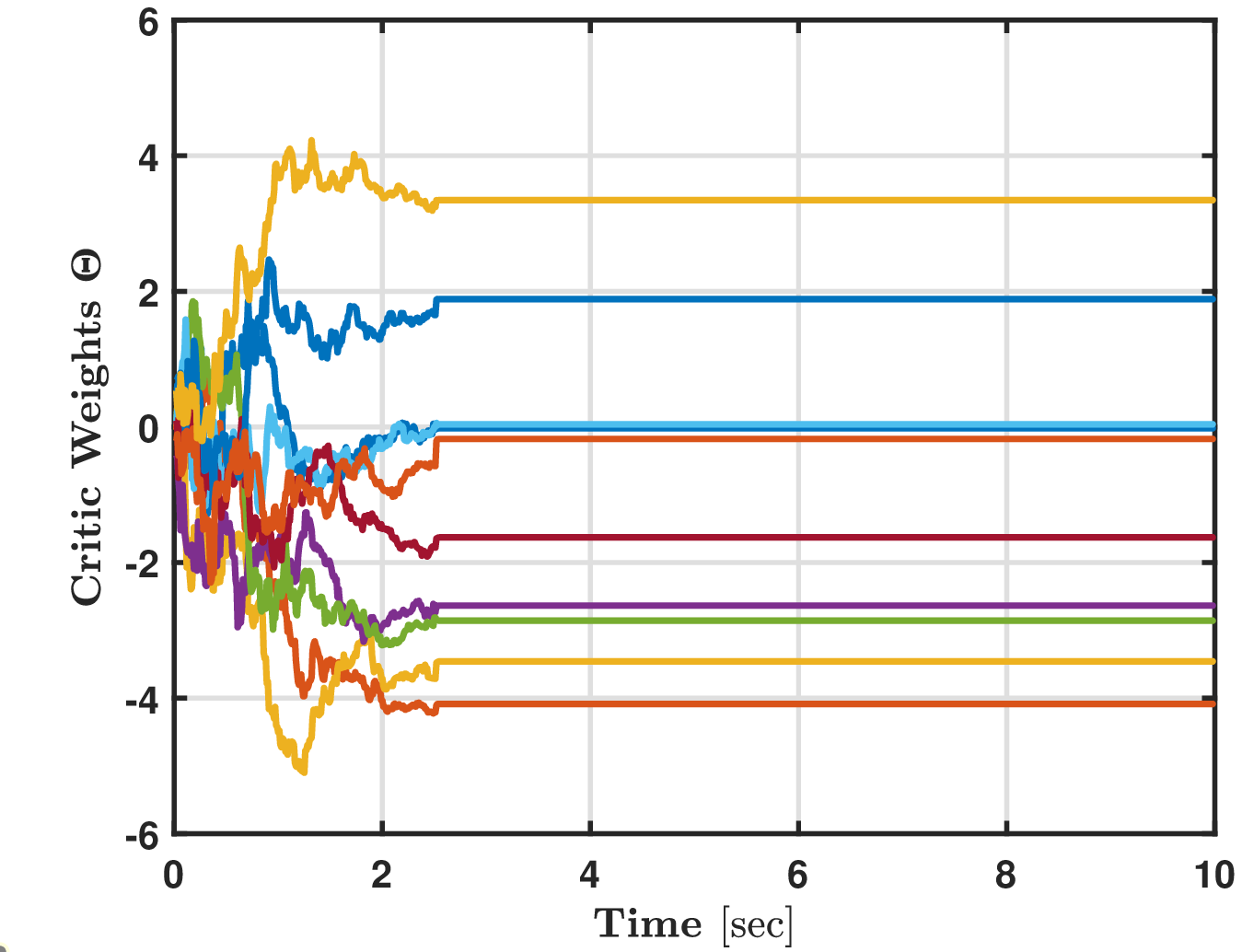}%
		}
		\subcaptionbox{The evolution of the control signal~${\bf u}_k$.%
			\label{fig:con2}}
		{%
			\includegraphics[width=0.33\textwidth]{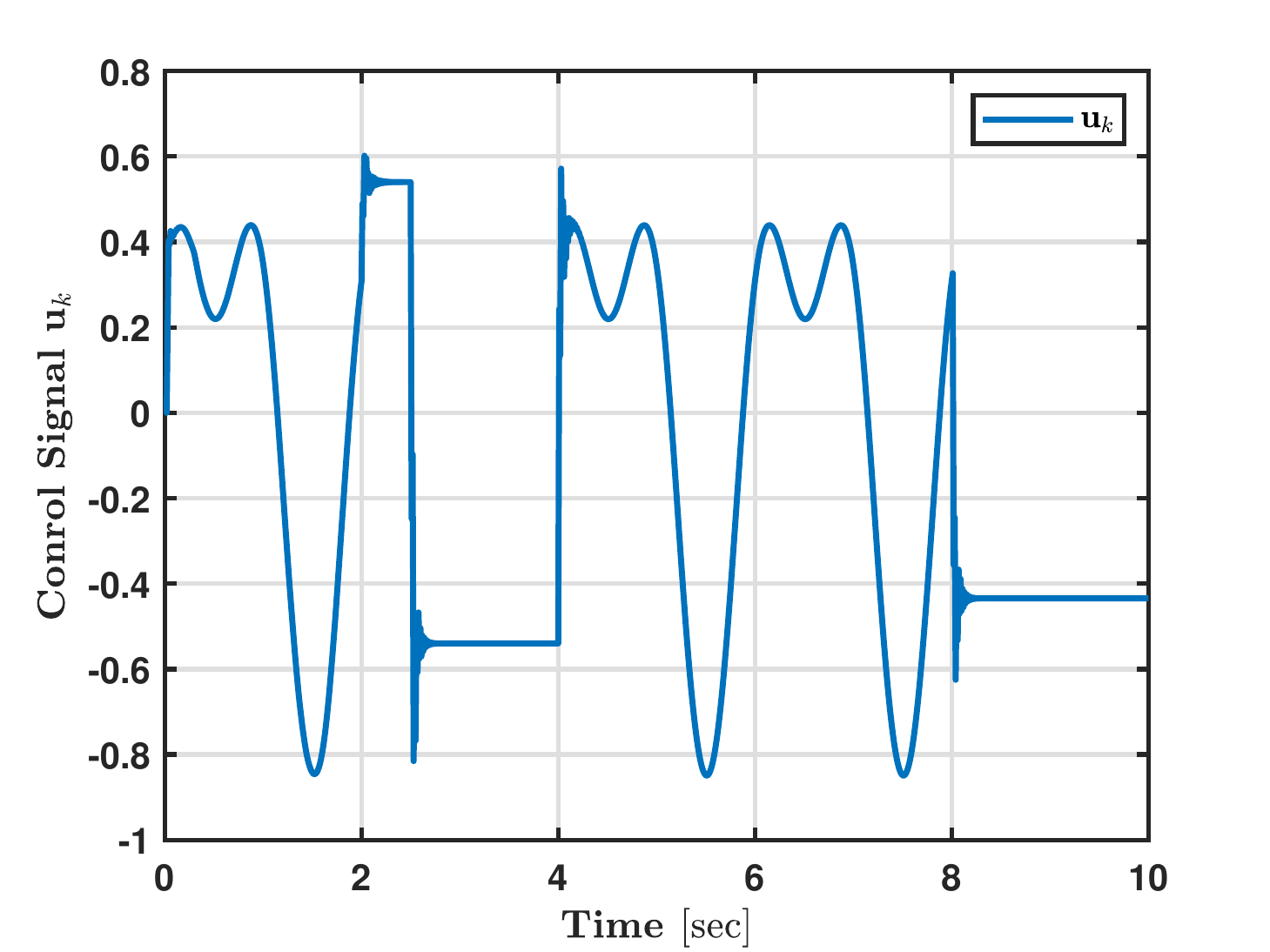}%
		}
		\hfill \hspace{-15pt}
		\\[2ex]
		\mbox{}\hfill
		\subcaptionbox{The evolution of the trajectory-tracking performance.%
			\label{fig:ref2}}
		{%
			\includegraphics[width=0.35\textwidth]{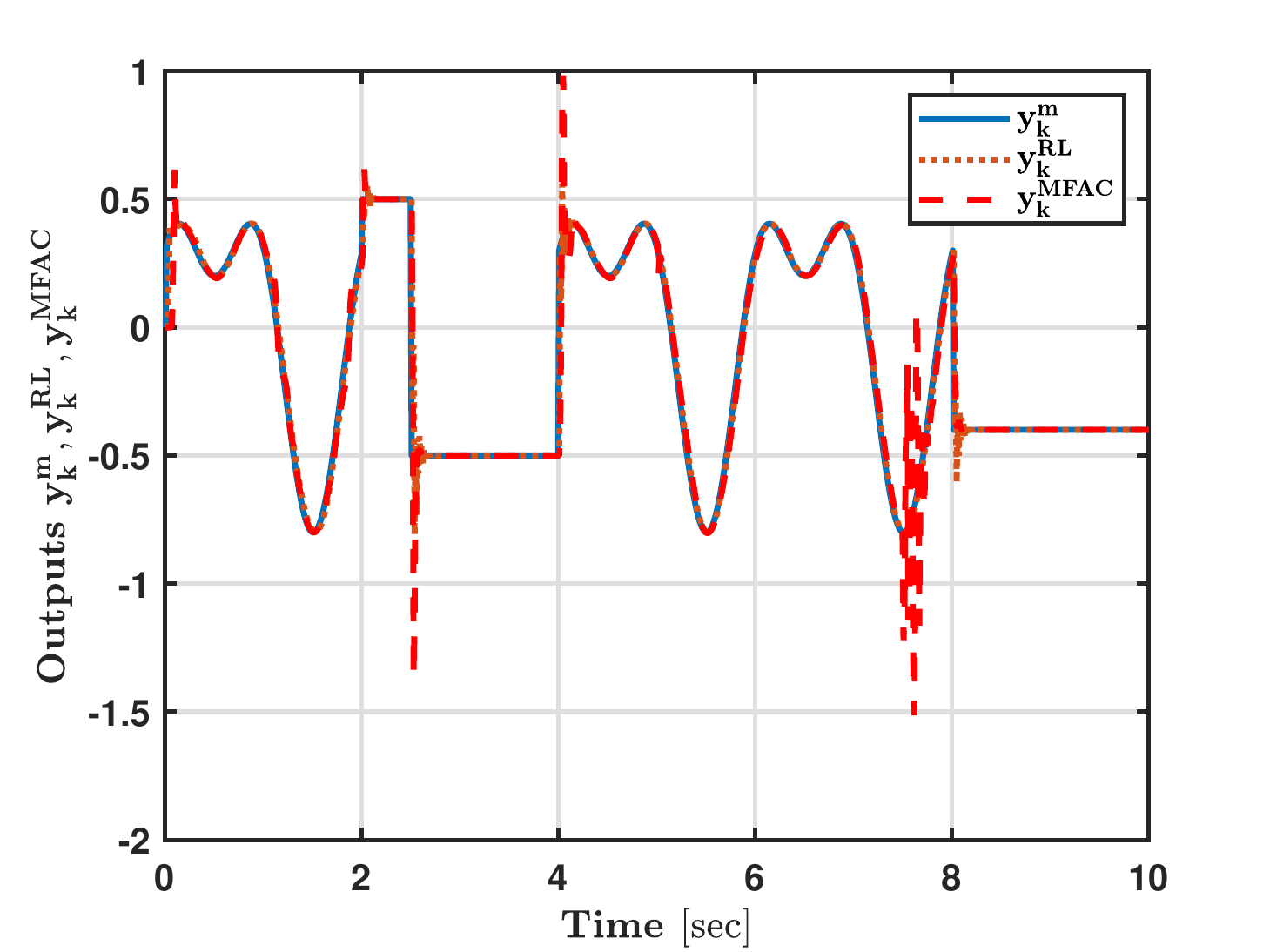}%
		}
		\hfill
		\subcaptionbox{The evolution of the error signal ${\bf \varepsilon}_k$.%
			\label{fig:err2}}
		{%
			\includegraphics[width=0.35\textwidth]{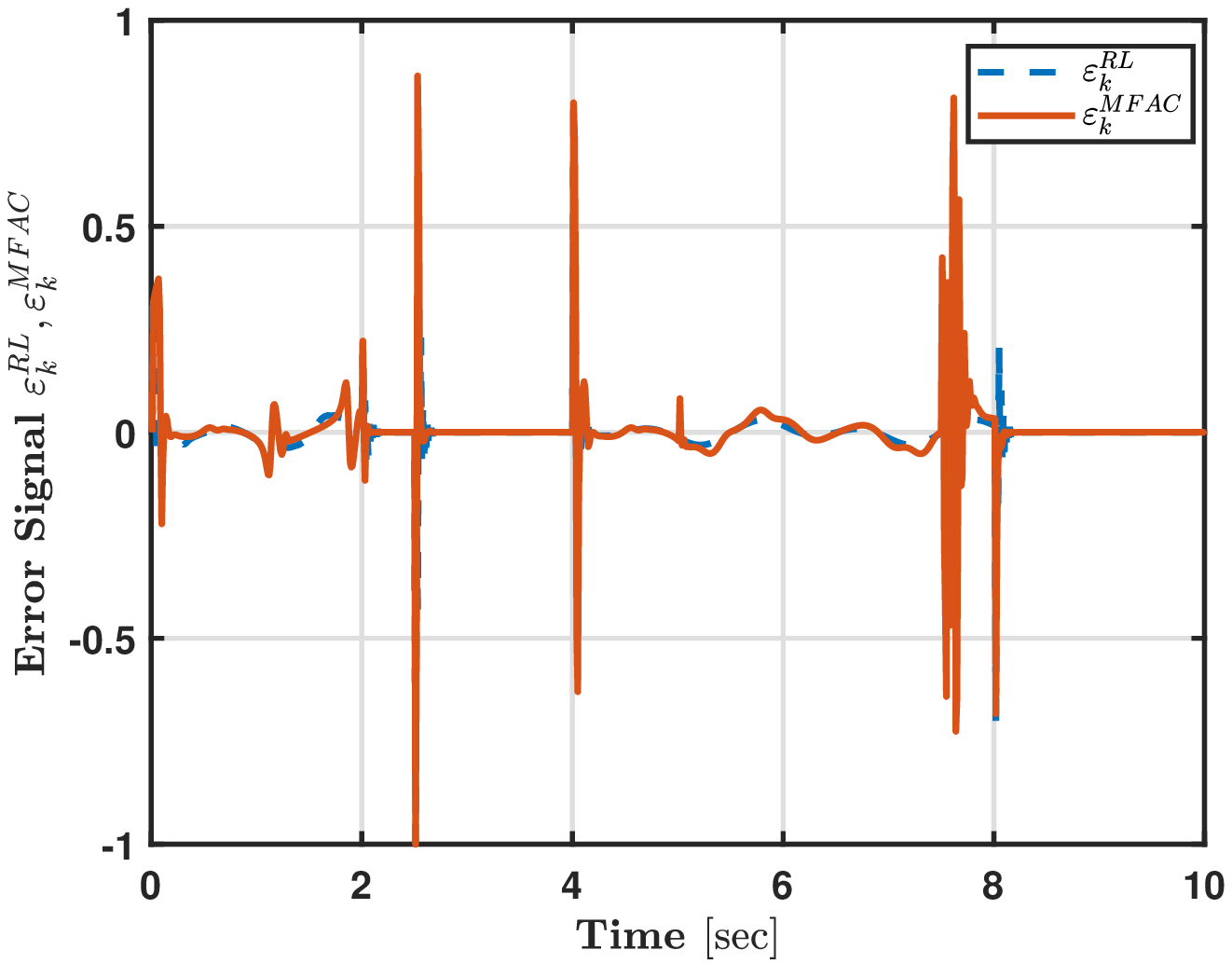}%
		}
		\hfill\mbox{}
		\caption{Simulation results of Case~2.\label{fig:fig2}} 
	\end{figure*}

	%
	%

	\section{Conclusion\label{sec:conclus}}
	The work combines an  RL approach with a projection-based adaptation mechanism to solve a model-reference adaptive control problem. This solution uses a moving finite-horizon of model-following error measurements. Further,the structure of the proposed model-following error vector reflects the order of the error dynamics. The adaptive strategy does not employ any explicit dynamic information of the process or the reference model. A PI technique is considered to solve the underlying Bellman equation. Finally, actor-critic approximation structures are designed to implement the PI solution, where the adaptation rules follow a Lagrange-based projection mechanism. Future research will extend the adaptive learning solution to a multi-agent setting.

	\balance

	\bibliographystyle{IEEEtran}
	\bibliography{bib_Neuro_Attit}
\end{document}